\documentclass[12pt]{article}

\usepackage[margin=1.25in]{geometry}
\usepackage{times}
\usepackage{mathptmx}
\usepackage{graphicx}
\usepackage{color}
\usepackage{ulem}
\usepackage{amsfonts,amsmath,amsthm,amssymb}
\usepackage{authblk}

\newtheorem{theorem}{Theorem}
\newtheorem{lemma}[theorem]{Lemma}
\newtheorem{prop}[theorem]{Proposition}
\newtheorem{corollary}[theorem]{Corollary}
\newtheorem{dfn}[theorem]{Definition}

\newcommand{\RR}{\mathbb{R}}
\newcommand{\RRn}{\RR^n}

\newcommand{\RP}{\mathbb{RP}}

\newcommand{\VV}{\mathcal{V}}

\newcommand{\OO}{\mathcal{O}}
\newcommand{\PP}{\mbox{colspace}}
\newcommand{\PPA}{\mbox{colspace}(A)}
\newcommand{\CC}{\mathcal{C}}
\newcommand{\LL}{\mathcal{L}}
\newcommand{\SSS}{\mathcal{S}}
\def\citep{\cite}

\usepackage{graphicx}

\title{Arbitrage and Geometry}

\author{Daniel Q. Naiman \& Edward R. Scheinerman\\ \small Department of Applied Mathematics and
Statistics\\Johns Hopkins University\\Baltimore, MD\\
daniel.naiman@jhu\\
ers@jhu.edu}

\date{\today}

\begin{document}
\normalem
\nocite{*}
\maketitle

\section{Introduction}
Arbitrage is the financial world's tooth fairy.  Just as the tooth
fairy provides a guaranteed financial return for baby teeth, an
\emph{arbitrage opportunity} is an investment, or combination of
investments, that is guaranteed to yield a profit. Arbitrage shares
another important property with the tooth fairy: neither of them exist.
The ``no free lunch'' assumption, that arbitrage opportunities in the
marketplace are unavailable, has played a fundamental role in
financial economics. In 1958, Modigliani and Miller
\citep{ModiglianiMiller} used the principle to argue that the way a
company finances itself, either by use of bonds or by issuing
additional stock is irrelevant in determining its value.  The
principle appears in the options pricing formula due to Merton, Black,
and Scholes \citep{BlackScholes,Merton}.  The principle was placed on
a solid theoretical footing in the work of Ross
\citep{Ross1976a,Ross1976b}. A gentle introduction to the arbitrage
principle is \citep{Varian}.

This article introduces the notion of arbitrage for a situation
involving a collection of investments and a payoff matrix describing
the return to an investor of each investment under each of a set of
possible scenarios. We explain the Arbitrage Theorem, discuss its
geometric meaning, and show its equivalence to Farkas' Lemma.  We then
ask a seemingly innocent question: given a random payoff matrix, what
is the probability of an arbitrage opportunity? This question leads to
some interesting geometry involving hyperplane arrangements and
related topics.

\subsection{Payoff matrices.}  A dizzying array of investment
opportunities is available to someone who has money to invest. There
are stocks, bonds, commodities, foreign currency, stock options,
futures, real estate, and of course, savings accounts.  We make the
simplifying assumption that there are $n$ possible choices of
investments available, and that at this moment of time (today) one may
invest any amount of one's finances in each. At the end of a fixed
time period (tomorrow), a single unit of cash (a dollar, a euro, a
yen, etc.)  invested in investment $i$ will be worth some amount of
money, so that changes in value of the investments can be represented
as an $n$-dimensional vector. A further simplifying assumption is that
there are finitely many mutually exclusive \emph{scenarios} that can
occur, and that each scenario leads to a specific gain (or loss) for
each investment.

As a consequence of these assumptions, we represent the changes in
value, or \emph{returns}, be they gains or losses, in an $m \times n$ \emph{payoff
  matrix} This is a matrix with a row corresponding to each scenario
and a column corresponding to each investment. The $j,i$ entry,
$a_{ji},$ gives the change in value at the end of the time period
based on a unit invested in investment $i$ under scenario $j$.  See
Figure~\ref{fig:payoff}.  Entries in this matrix are positive if the
investment gains value and negative if it loses value.

\subsection{Risk-free Rate.}  It is typical to assume that there is a
\emph{risk-free} investment (e.g. U.S. Treasury Bills) available to
the investor that yields a payoff tomorrow of $1+r$ units for each
unit invested.  Thus, a single unit today is guaranteed to be worth
$1+r$ units tomorrow.  We express gains or losses tomorrow in
\emph{present day} units, by multiplying by a \emph{discount factor}
of $1/(1+r).$

Thus, if we invest one unit of cash today in an investment and it
becomes worth $x$ units tomorrow, the relevant entry in the payoff
matrix is its \emph{present value}  $-1+x/(1+r)$.

\begin{center}
\begin{figure}[ht]
$$
\begin{array}{c|ccccccc}
     & I_1 & I_2 & \cdots & I_i & \cdots & I_n \\ \hline
\mbox{Scenario}_1 & a_{11} & a_{12} & \cdots & a_{1i} & \cdots & a_{1n} \\
\mbox{Scenario}_2 & a_{21} & a_{22} & \cdots& a_{2i} & \cdots&  a_{2n}\\
\vdots & \vdots & \vdots & &\vdots &&\vdots  \\
\mbox{Scenario}_j & a_{j1} & a_{j2} & \cdots& a_{ji} & \cdots&  a_{jn}\\
\vdots & \vdots & \vdots & &\vdots &&\vdots  \\
\mbox{Scenario}_m & a_{m1} & a_{m2} & \cdots& a_{mi} & \cdots&  a_{mn}\\
\end{array}
$$
\caption{The $j,i$ entry of \emph{payoff matrix} $A$ is the
  gain (or loss) tomorrow under scenario $j$ for a unit invested in
  investment $i$ today.}
\label{fig:payoff}
\end{figure}
\end{center}

\subsection{Example: The Bernoulli model}  The following simple example
is a standard one that appears in most introductory
books on options pricing, and is included here
to give the reader a taste for how options are priced. Assume that
when we invest one unit in a certain stock today, currently priced at
$S$, its value tomorrow either be one of two possibilities, either
$Su$ or $Sd$, where $d < 1+r < u$, so that the stock under- or
over-performs the risk-free rate.  Consequently, there are two
possible payoffs tomorrow of a single unit investment in the stock
today, $u$ or $d$, and discounting gives either a value change of
$-1+u/(1+r)$ or $-1+d/(1+r)$ in today's terms.

Another possibility is to purchase a stock option referred to as a
\emph{call}: we pay someone an amount $P$ today for the right to buy
from them a share of the stock at a predetermined price $K$ (called
the \emph{strike} price) tomorrow (no matter what happens to the
market price), where $Sd < K < Su.$ If the stock goes up, then we
would exercise our right to buy it at $K$ and then immediately sell it
at the market price of $Su,$ instantly netting $Su-K.$ In this case,
the net gain for buying the option (taking into account discounting)
is
$$
-P + (Su-K)/(1+r).
$$
In other words, investing a single unit of money in such an option
results in a value change of
$$
-1+(Su-K)/(P(1+r)).
$$
On the other hand, if the stock goes down we would not want to
exercise the option since we would be better off purchasing the
stock at the lower market price, and in this case we have simply lost
the cost of the option.

Thus, we can write down a $2 \times 2$ payoff matrix corresponding to
the two investment opportunities (stock, option) and the two possible
scenarios (stock goes up, stock goes down) as in
Figure~\ref{fig:payoff-stock-option}.
A third column representing the payoff for a risk-free investment of
one unit could be adjoined to this matrix. However, this column would consist of zeros,
since discounting leaves the present value of such an investment unchanged.

\begin{center}
\begin{figure}[ht]
$$
\begin{array}{c|cc}
                & \mbox{stock} & \mbox{option} \\ \hline
\mbox{up}   & -1+u/(1+r)  & -1+(Su-K)/(P(1+r)) \\
\mbox{down} & -1+d/(1+r)& -1 \\
\end{array}
$$
\caption{The payoff matrix for a stock whose price today is $S$ and whose
  price tomorrow takes one of two values $Su$ or $Sd$, and for an
  option whose price today is $P$ and gives one the right to purchase
  the stock at a strike price of $K$ tomorrow.  The risk-free interest
  rate for this one day period is $r$.}
\label{fig:payoff-stock-option}
\end{figure}
\end{center}

\subsection{Arbitrage}  In finance, an \emph{arbitrage} opportunity is a
combination of investments whose return is guaranteed to outperform
the risk-free rate.  If such an opportunity were to exist we could
borrow money at the lower rate, get the higher return, repay our debt,
and pocket the difference. Since our initial investment is unlimited,
so is the amount of money we could make.  In modeling the behavior of
markets, a common approach is to assume that an equilibrium condition
exists under which such opportunities are not available; if they were
to become available, they could only exist for an extremely brief
period of time. Almost as soon as they are discovered, they are wiped
out as a result of the response of investors.

The no arbitrage assumption is as fundamental a principle for finance
as Newton's first and second laws of motion are for physics.  Another
important principle bears semantic resemblance to Newton's third law:
for every action there is an equal and opposite reaction.  In finance,
everything that can be bought can also be sold.

A consequence of this is that for every available payoff column there
is an investment opportunity that achieves exactly the opposite
effect, that is, changes the signs of the payoffs. This may seem
counterintuitive. For example, what investment leads to a payoff that
is complementary to the purchase of a share of stock?  The answer is
to \emph{short sell it}, that is, in essence, borrow it, sell it, and
later buy it back and subsequently return it to the
lender.\footnote{In fact, short-selling typically requires that one
  put aside assets as collateral, and practically speaking, investors
  are constrained by their actions.}  Consequently, a column in the
payoff matrix may be replaced by any nonzero multiple of the column,
and the modified matrix represents the same investment opportunities.

Given a specific payoff matrix $A$ that gives the behavior of all
possible investments and scenarios, an arbitrage opportunity is said
to exist if there is a linear combination of columns of $A$ all of
whose entries are strictly positive. Thus, an arbitrage opportunity is
a combination of buys and sells of the basic investments that yields a
net gain under all scenarios.

The absence-of-arbitrage assumption leads to constraints on payoff
matrices. To illustrate this, consider the stock/option example above.
If the two columns of the matrix in
Figure~\ref{fig:payoff-stock-option} are not scalar multiples of each
other, then they form a basis of $\RR^2$ and so we can find a
combination of buys and sells of the two investments to yield any
payoff vector we wish. Thus, the absence of arbitrage implies that the
two columns are multiples of each other, and this gives
$$
P=\frac{(Su-K)(1+r-d)}{ (1+r)(u-d)},
$$
so that once the various parameters involved, including the
interest rate, the possible factors by which the stock price can
change, the price of the stock today, and the strike price $K,$ are
set, the price of the option can be calculated.

\section{The Arbitrage Theorem}

The Arbitrage Theorem provides an interesting and convenient
characterization of the no arbitrage condition.  Given an $m \times n$
payoff matrix $A,$ and an $n$-vector $x$, the product $Ax$ gives the
payoff vector that results from investing $x_i$ in investment $i,$ for
$i=1,\ldots,n.$ Consequently, for a given payoff matrix $A,$ its column space \PPA,
represents the set of payoff vectors corresponding to all
possible combinations of investments.
This forms a subspace of $\RR^m$ of dimension at most $n$.

\begin{theorem}[Arbitrage]
  Given an $m \times n$ payoff matrix $A$, exactly one of the
  following statements holds:
  \begin{itemize}
  \item[(A1)] Some payoff vector in $\PP(A)$ has all positive
    components, i.e. $Av>0$ for some $v \in \RR^n.$
  \item[(A2)] There exists a probability vector $\pi =
    [\pi_1,\ldots,\pi_m]^t$ that is orthogonal to every column of $A,$
    i.e. $\pi^t A=0,$ where $\pi^t \geq 0$ and $\pi^t 1=1.$
  \end{itemize}
\end{theorem}

In other words, the absence of arbitrage is equivalent to the
existence of an assignment of probabilities to scenarios under which
every investment has an expected return of zero.

For a concrete example, consider the situation described above in
which our payoff matrix contains a column of the form
$$
\left[
\begin{array}{cc}
-1+u/(1+r)\\
-1+d/(1+r)\\
\end{array}
\right].
$$
The absence of arbitrage means that there is a probability vector
$$
\left[
\begin{array}{cc}
\pi_u\\
\pi_d\\
\end{array}
\right]
$$
orthogonal to it and every other column of the payoff matrix.  It
follows immediately that $\pi_u = \frac{1+r-d}{u-d}$ and $\pi_d =
1-\pi_u = \frac{u-1-r}{u-d}$.

These probabilities do not have the familiar interpretation as
relative frequencies.  However, the probability assignment is a
practical one in that under the assumption of no arbitrage, any
investment should have a zero expected payoff.  Thus, these
probabilities provide us with a method for establishing the price of
any security whose payoff depends on the behavior of the stock.  For
example, suppose we have the opportunity to invest in a security whose
price is $P$ today and whose payoff tomorrow is $\rho_u$ if the stock
goes up and $\rho_d$ if the stock goes down.  Then we can add a column
to our payoff matrix describing the present value corresponding to a
unit investment
$$
\left[
\begin{array}{cc}
-1+\rho_u/P(1+r)\\
-1+\rho_d/P(1+r).\\
\end{array}
\right].
$$
Absence of arbitrage leads to the conclusion that
$$
\pi_u \left(-1+\rho_u/P(1+r) \right) + \pi_d\left(-1+\rho_d/P(1+r) \right)=0,
$$
and solving for $P$ we obtain
$$
P=\pi_u \rho_u/(1+r) + \pi_d \rho_d/(1+r),
$$
that is, the no-arbitrage price of the investment is the
\emph{expected
value} of its discounted payoff.

The Arbitrage Theorem  (also known as \emph{Gordon's Theorem}) is
but one example of a \emph{theorem of the
  alternative} appearing in convex analysis in which one asserts the
existence of a vector satisfying exactly one of two properties.  One
of the more fundamental results of this type is Farkas' Lemma
\citep{Farkas}, treatments of which may be found in many books on
finite-dimensional optimization, including
\citep{BT,Padberg,StoerWitzgall} (see also \citep{AvisKaluzny}).  The
Arbitrage Theorem is frequently presented as a simple consequence of
Farkas' Lemma, for example, see \citep{BT,RossSheldon}.
A lengthy treatment of theorems of the alternative appears in
\citep{mangasarian}, together with a table (Table 2.4.1) of eleven such theorems.

\begin{lemma}[Farkas]
Given an $m \times n$ matrix $A$ and an
$n$-vector $b,$ exactly one of the following statements holds:
\begin{itemize}
\item[(i)]
There exists $x \geq 0$ such that $Ax=b.$
\item[(ii)]
There exists $y$ such that $y^tA\geq 0$ and $y^tb< 0.$
\end{itemize}
Equivalently, the following two statements
are equivalent:
\begin{itemize}
\item[(F1)]
There exists $x \geq 0$ such that $Ax=b.$
\item[(F2)]
$y^tA\geq 0$ implies $y^tb\geq 0.$
\end{itemize}
\end{lemma}

\begin{dfn}
The \emph{polyhedral convex cone} generated by
$a^{(i)}\in \RR^m,~i=1,\ldots,n$ is the subset of $\RR^m$ defined by
$$
\left\{ \sum_{i=1}^n x_i a^{(i)} ~:~ x_i \geq 0 \right\}.
$$
We denote this set by $\CC(a^{(i)},i=1,\ldots,n)$ and we
refer to the $a^{(i)},~i=1,\ldots,n$ as \emph{generators} of the cone.
  \end{dfn}

Farkas' Lemma has an intuitive geometric interpretation. Let the columns
of $A$ be denoted by $a^{(i)},~i=1,\ldots,n,$ and
take
$$
\CC_A := \CC(a^{(i)},i=1,\ldots,n).
$$
Condition (F1) says that $b \in \CC_A.$ On the other hand, assuming
$b \neq 0,$ (F2) says that if $y$ makes a non-obtuse angle with every
non-zero cone generator $a^{(i)}$ then the angle $y$ makes with
$b$ is non-obtuse.

We show how the Arbitrage Theorem follows from Farkas' Lemma and,
conversely, how to prove Farkas' Lemma from the Arbitrage
Theorem. Farkas' Lemma is central to the theory of linear programming
and, in the same spirit, \cite{reichmeider} shows how various
combinatorial duality theorems can be derived from one another.

\subsection{Proof of the Arbitrage Theorem using Farkas' Lemma}

\begin{proof}
First, if (A1) and (A2) are both satisfied, then we have
$$
\pi^t A v = 0v =0,
$$
but on the other hand $Av>0$ so $\pi^t Av >0$ since the entries in $\pi$
are nonnegative and sum to one. It follows that (A1) and (A2) are mutually
exclusive.

Next we show that (A1) and (A2) cannot both fail. If (A2) fails, then there is
no solution to
$$
\label{no_solution_equation}
\left[
\begin{array}{c}
A^t \\
1\\
\end{array}
\right] \pi =
\left[
\begin{array}{c}
0 \\
1\\
\end{array}
\right]
$$
with $\pi \geq 0.$
It follows that condition (i)
in Farkas' Lemma fails for
$\tilde{A} =\left[
\begin{array}{c}
A^t \\
1\\
\end{array}
\right]$
and
$\tilde{b} =\left[
\begin{array}{c}
0 \\
1\\
\end{array}
\right],$
so Farkas' Lemma guarantees that (ii) holds.
Therefore, there exists an $n+1$-vector $y$ such that
$y^t \tilde{A}\geq 0$
and
$y^t \tilde{b} < 0.$
Writing
$
y = \left[
\begin{array}{c}
v \\
s\\
\end{array}
\right]
$
where $v = [v_1,\ldots,v_n]^t$ and $s \in \RR$ these conditions say that
$Av + s1 \geq 0$ and
$s <0,$ and we conclude that $Av >0,$ so (A1) is valid.
\end{proof}

\subsection{Proving Farkas's Lemma from the Arbitrage Theorem}
As it turns out, the Arbitrage Theorem together, with some work, leads
to a proof of Farkas' Lemma.  Note first that the implication (F1)
$\Rightarrow$ (F2) is immediate: if we can find $x \geq 0$ such that
$Ax=b,$ then assuming that $y^t A \geq 0$ we have $y^t b = y^t Ax \geq
0.$

Our focus is on the more difficult implication (F2) $\Rightarrow$ (F1).
It is instructive to first prove this under a certain technical assumption.
For a polyhedral convex cone
$\CC \subseteq \RR^m,$ we define
$$
\LL(\CC) := \CC \cap -\CC.
$$
This is easily seen to form a subspace of $\RR^m$ and consists of
all lines completely contained in $\CC$ passing through the origin.

\begin{dfn}
  A polyhedral convex cone $\CC$ is \emph{pointed} if
  $\LL(\CC)=\{0\}.$ In other words, if it contains no lines through
  the origin, that is, $v, -v \in \CC$ implies $v=0.$
\end{dfn}

The reader is referred to Figures \ref{fig:pointed-cone} and \ref{fig:nonpointed-cone}.

\begin{figure}[ht]
\scalebox{.7}{
\begin{minipage}{\textwidth}
{\input{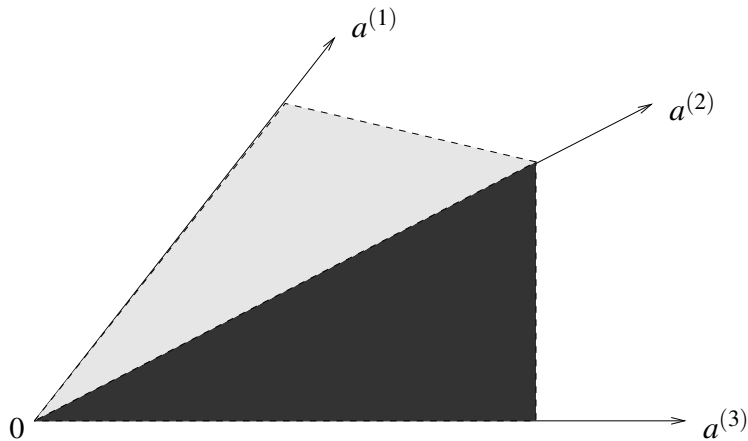}}
\end{minipage}}
  \caption{The vectors $a^{(i)},i=1,2,3$ generate a pointed cone in
    $\RR^3.$}
  \label{fig:pointed-cone}
\end{figure}

\begin{figure}[ht]
\scalebox{.7}{
\begin{minipage}{\textwidth}
{\input{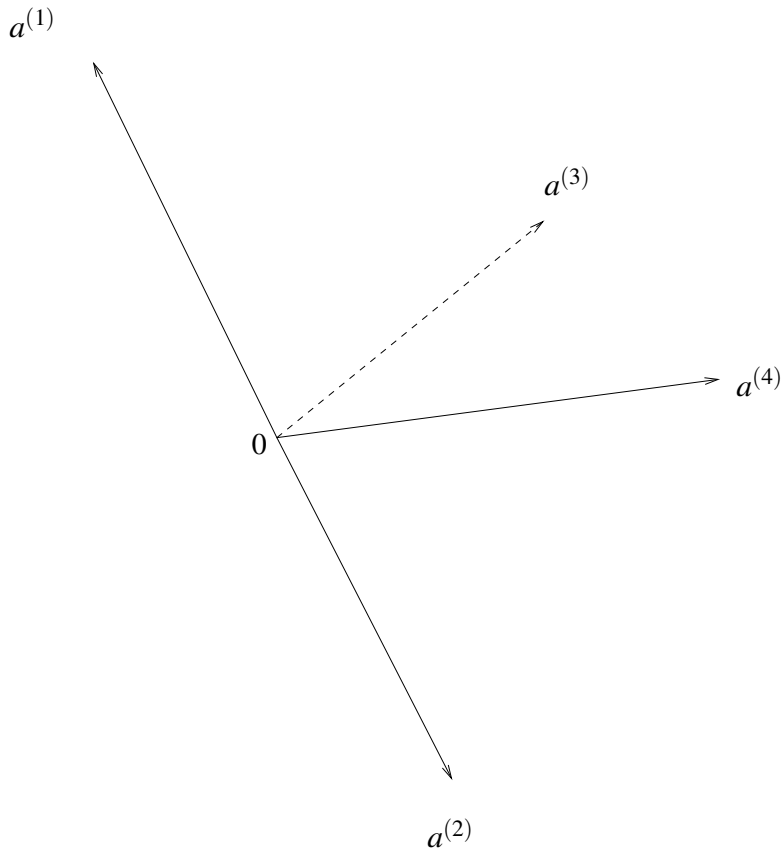}}
\end{minipage}}
  \caption{The vectors $a^{(i)},i=1,2,3,4$ generate a non-pointed cone in
    $\RR^3.$}
  \label{fig:nonpointed-cone}
\end{figure}

It is an elementary exercise to prove that pointedness for the cone
$\CC_A$ says that whenever
$\sum_{i=1}^n x_i a^{(i)} = 0$ for $x_i \geq 0,~i=1,\ldots,n$ we can conclude that
$x_i a^{(i)}=0$ for all $i.$

\begin{lemma}
  \label{Farkas_for_pointed_cones_lemma}
  The implication (F2) $\Rightarrow$ (F1) holds when $A$ is an $m
  \times n$ matrix such that $\CC_A$ is a pointed cone.
\end{lemma}

\begin{proof}
  Observe that (F2) holds for a matrix $A,$ then it also holds for the
  matrix obtained by removing any zero columns of $A.$ Also, if (F1)
  holds for the matrix with the zero columns of $A$ removed then it
  also holds for the original matrix. It follows that, without loss of
  generality, we may assume that all of the columns of $A$ are
  nonzero.

  It follows from (F2) that
  $$
  \left[ \begin{array}{c} A^t \\ -b^t \end{array} \right]y \not > 0,
  $$
  for all $y.$ Thus, condition (A1) fails, and from the Arbitrage
  Theorem we conclude that (A2) holds, that is, there exists a
  probability vector $\left[ \begin{array}{c} u \\ s \end{array}
  \right]$ with $u$ an $n$-vector, and $s$ a scalar, such that
  $$
  [u^t,s] \left[ \begin{array}{c} A^t \\ -b^t \end{array} \right]=0.
  $$
  Expanding and transposing, we obtain $Au = s b.$ We can assume
  $s>0$ since otherwise the assumption of pointedness is violated.  It
  follows that we can divide both sides by $s$ and we obtain (F1).
\end{proof}

To establish the general case we need to collect some basic
results concerning the structure of non-pointed polyhedral convex cones.
The basic idea is that one can always represent a cone as a direct
sum of a linear subspace and a pointed \emph{slice} of the cone (see Figure \ref{fig:slicing-a-cone}).
The proof is left as an elementary exercise.

\begin{lemma}
  \label{cone_decomposition_lemma}
  Given a polyhedral convex cone $\CC = \CC(a^{(i)},i=1,\ldots,n)
  \subseteq \RR^m,$ let $\LL=\LL(\CC)$ and let $\LL^\perp$ denote its
  orthogonal complement. Then any given $x\in \CC$ can be expressed
  uniquely as $x = u + v$ where $u \in \CC \cap \LL^\perp,$ $v \in
  \LL.$ We express this statement symbolically by writing
  $$
  \CC = (\CC\cap \LL^\perp)  \oplus \LL.
  $$
  Furthermore, if
  $\tilde{a}^{(i)}$
  denotes the orthogonal projection of $a^{(i)}$ onto
  $\LL^\perp,$ for $i=1,\ldots,n,$ then
  $$
  \CC\cap \LL^\perp = \CC( \tilde{a}^{(i)},i=1,\ldots,n),
  $$
  and the slice $\CC\cap \LL^\perp$ is a pointed polyhedral convex cone.
\end{lemma}

\begin{figure}[ht]
\scalebox{.7}{
\begin{minipage}{\textwidth}
{\input{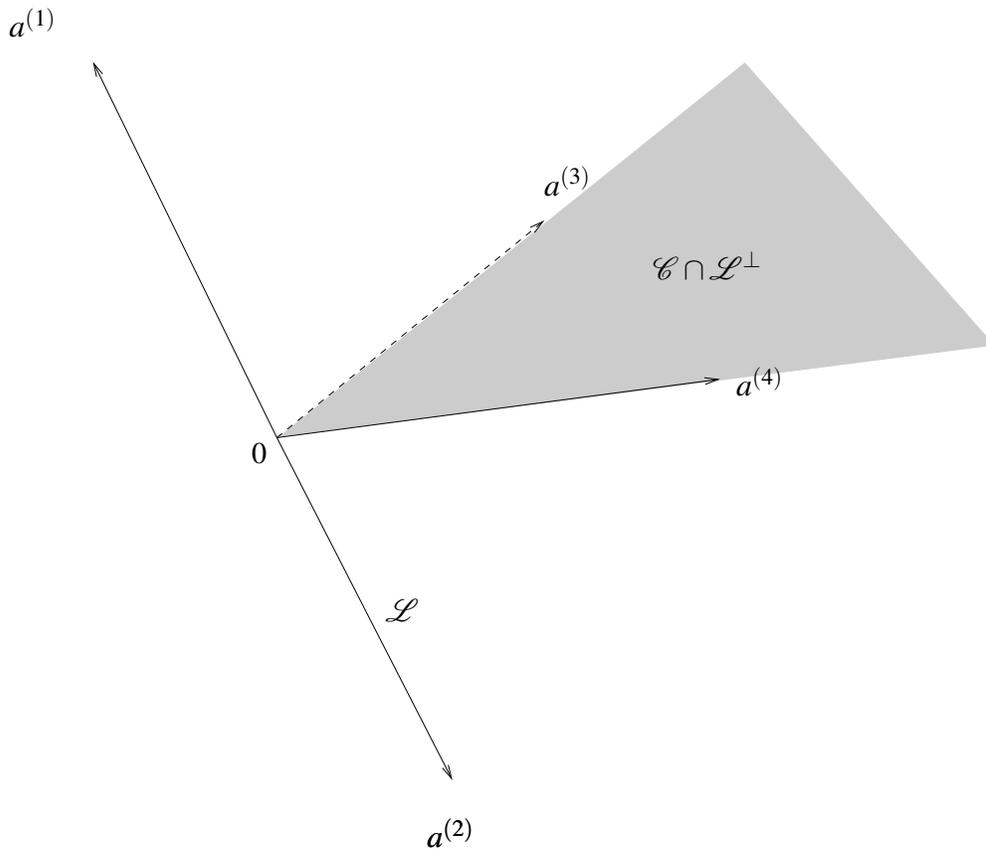}}
\end{minipage}}
  \caption{Slicing the cone in Figure~\ref{fig:nonpointed-cone} yields a pointed cone.}
  \label{fig:slicing-a-cone}
\end{figure}

Armed with these results we are now prepared to use the Arbitrage
Theorem to show (F2) $\Rightarrow$ (F1).

\begin{proof}[Proof of Farkas' Lemma by way of the Arbitrage Theorem]
  Assume (F2) holds for a given $m \times n$ matrix $A$ with columns
  $a^{(i)},i=1,\ldots,n.$ Take $\CC = \CC(a^{(i)},i=1,\ldots,n)$,
  $\LL=\LL(\CC),$ $\LL^\perp,$ and $\tilde{a}^{(i)}$ to be the orthogonal projection
  of $a^{(i)})$ on $\LL^\perp,$ for $i=1,\ldots,n.$
  We can write $a^{(i)} =
  \tilde{a}^{(i)} + \hat{a}^{(i)},$ and where $\hat{a}^{(i)} \in \LL.$
  Also, we define $\tilde{b}$ to be the orthogonal projection of
  $b$ onto $\LL^\perp$ and we write $b =
  \tilde{b} + \hat{b},$ where $\hat{b} \in \LL.$

  Now suppose $y \in \LL^\perp \subseteq \RR^m$ and $y^t
  \tilde{a}^{(i)}\geq 0$ for $i=1,\ldots,n.$ Then $y^t a^{(i)} = y^t
  \tilde{a}^{(i)} + y^t \hat{a}^{(i)} = y^t \tilde{a}^{(i)}\geq 0,$
  for $i=1,\ldots,n$ and we can conclude from (F2) that $y^t b \geq
  0.$ It follows that $y^t\tilde{b} = y^t(b - \hat{b}) = y^t b \geq
  0.$

  Summarizing, we have shown that
  $$
  y \in \LL^\perp \subseteq \RR^m
  \text{ and }
  y^t
  \tilde{a}^{(i)}\geq 0 \text{ for }
  i=1,\ldots,n \quad\Rightarrow\quad
  y^t \tilde{b} \geq 0.
  \leqno{(\text{F}2')}
  $$

  Since the $\tilde{a}^{(i)},i=1,\ldots,n$ generate a pointed cone in
  $\LL^\perp$ and $\tilde{b} \in \LL^\perp,$ Lemma~\ref{cone_decomposition_lemma}
  allows us to conclude that
  $\tilde{b} \in \CC(\tilde{a}^{(i)},i=1,\ldots,n),$ i.e.  we can
  write $\tilde{b} = \sum_{i=1}^n x_i \tilde{a}^{(i)}$ for some $x_i
  \geq 0.$ It follows that
  $$
  b = \tilde{b} + \hat{b} =\sum_{i=1}^n x_i \tilde{a}^{(i)}+\hat{b}
  = \sum_{i=1}^n x_i (a^{(i)}+\hat{a}^{(i)}) +\hat{b} = \sum_{i=1}^n
  x_i a^{(i)}+\left\{ \sum_{i=1}^n x_i \hat{a}^{(i)}
    +\hat{b}\right\}\in \CC,
  $$
  where we have used the fact that $\LL \subseteq \CC$ to conclude
  that the term in braces lies in $\CC.$
\end{proof}

\section{Geometry of generic payoff matrices}

For a given $m\times n$ payoff matrix, the existence of arbitrage
amounts to the statement that $\PPA$ intersects the positive orthant
$\OO^+ := \left\{ x \in \RR^m : x_j>0,~j=1,\ldots,m\right\}.$ Thus the
probability that a random payoff matrix exhibits arbitrage equals the
probability that its column space meets the positive orthant which, in
turn, is the probability that a random $n$-dimensional subspace of
$\RR^m$ intersects $\OO^+$.

This leads us to ask: How many orthants does an $n$-dimensional
subspace of $\RR^m$ intersect? For the most part, the answer is
independent of the choice of the subspace and is closely related to
other geometric counting problems.
See \cite{buck} in particular, but also
\cite{arrangements2000},
\cite{Brualdi-book}~pages~285--290,
\cite{Comet1974}~page~73,
\cite{cheese-solution},
\cite{may+smith},
\cite{orlik},
\cite{weisstein:plane-division},
\cite{monthly-cheese},
\cite{zaslavsky}, and
\cite{zaslavsky-mem}.

Let $e_1,e_2,\ldots,e_m$ denote the standard orthonormal basis of $\RR^m$.

\begin{dfn}
  Let $\VV$ be an $n$-dimensional subspace of $\RR^m$. We say that
  $\VV$ is \emph{generic} if for some (and therefore for any) basis
  $\{v_1,v_2,\ldots,v_n\}$ of $\VV$ and for any $m-n$ standard basis
  vectors $e_{i_1},\ldots,e_{i_{m-n}}$, the vectors
  $\{v_1,v_2,\ldots,v_n, e_{i_1},\ldots,e_{i_{m-n}}\}$ are linearly
  independent.  Likewise, an $m\times n$ matrix $A$ is called
  \emph{generic} if  $\PPA$ is a generic subspace of
  $\RR^m$.
\end{dfn}

\medbreak

Next we present some notation for the orthants of $\RR^m$.
Let $\Pi_j = e_j^\perp$. The subspaces $\Pi_j$ are the coordinate
hyperplanes and these separate $\RR^m$ into orthants, i.e., the
connected components of $\RR^m - \bigcup \Pi_j$. Two vectors $v$ and
$w$ (none of whose coordinates is zero) are in the same orthant
provided the sign of $v_i$ equals the sign of $w_i$ for all $i$.

Denote by $\SSS_m$ the set of $m$-vectors $\delta =
(\delta_1,\ldots,\delta_m)$ where each $\delta_i$ is $\pm1$.
Multiplication of members of $\SSS_m$ is defined coordinatewise.
The orthant $\OO_\delta$ is defined by
$$
\OO_\delta = \{x \in \RR^m : \delta_i x_i > 0,~i=1,\ldots,m \}.
$$
The positive orthant $\OO^+$ is simply $\OO_{(1,1,\ldots,1)}$.

\medbreak

A generic subspace $\VV$ of $\RR^m$ intersects some subset of the
orthants of $\RR^m$. Two points in $\VV$ lie in different orthants of
$\RR^m$ exactly when they are separated by some coordinate
hyperplane(s) $\Pi_j$. Thus, there is a one-to-one correspondence
between the orthants intersected by $\VV$ and the connected components
of
$$
\VV - \bigcup_{j=1}^m \Pi_j .
$$

The intersections of $\VV$ with the coordinate hyperplanes $\Pi_j$ are
subspaces of $\VV$ with particular properties; we show that they have
codimension~$1$ (i.e., have dimension $n-1$) and lie in general
position.

\begin{dfn}
  Subspaces $H_1,\ldots,H_m$ of codimension 1 in an $n$-dimensional
  vector space are said to be in \emph{general position}
  provided that $\dim\left( \bigcap_{j \in J} H_j\right)=n-\vert
  J\vert$ for all $J \subseteq \{1,\ldots,m\}$ with $1\leq \vert J
  \vert \leq n.$
\end{dfn}

For example, when $n=2$, any collection of distinct lines through the
origin are in general position. When $n=3$, a collection of distinct planes
through the origin are in general position provided no three of them
intersect in a line.

For a subspace $H$ of codimension 1 in a  vector space  $\VV$, the
complement $\VV - H$ consists of a pair of half-spaces that we can label
arbitrarily as $H^+$ and $H^-.$

Given $m$ subspaces $H_1,\ldots,H_m$ of codimension 1 whose associated
half-spaces have been labeled, we can assign to every point in the complement $\VV -
\bigcup_{i=1}^m H_i$ an $m$-vector
$\delta(x) = (\delta_1(x),\ldots,\delta_m(x))$ of signs, where the $\delta_i(x)$ indicates
which half-space, $H_i^+$ or $H_i^-,$ $x$ lies in. For each sign vector
$\delta \in \SSS_m$ the set $\bigcap_{i=1}^m \left\{ x \in \VV ~:~ x\in
H_i^{\delta_i} \right\}$ is an intersection of open half-spaces, so it is either empty
or it is the interior of a convex polyhedron. Thus, the
$H_1,\ldots,H_m$ decompose $\VV - \bigcup_{i=1}^m H_i$ into connected components,
each associated with some $\delta \in \SSS_m$, which we refer to as
\emph{cells}.

\begin{lemma}
  \label{generic_implies_general_position_lemma}
  Let $A$ be a generic $m \times n$ matrix and define $H_i = \PPA
  \cap \Pi_i$ for $i=1,\ldots,m.$ Then $H_1,\ldots,H_m$ are subspaces
  of codimension~$1$ in general position in $\PPA.$ Furthermore, the
  connected components of $\PPA - \bigcup_{i=1}^m H_i$ correspond to
  the orthants that $\PPA$ intersects.
\end{lemma}

\begin{proof}
  Let $k=\dim(\PPA)$ and
  let $J \subseteq \{ 1,\ldots,m\}$ with $1\leq \vert J \vert \leq k,$
  then $\vert J^c \vert \geq m-k$ so we can find distinct indices
  $i_1,\ldots,i_{m-k} \in J^c.$ Thus
  $$
  \bigcap_{j\in J} \Pi_j = \mbox{span} \{ e_i~:~ i \in J^c \}
  \supseteq \mbox{span} \{ e_{i_1}, \ldots, e_{i_{m-k}} \},
  $$
  and consequently
  $$
  \PPA + \bigcap_{j\in J} \Pi_j \supseteq \mbox{colspace}\left[
    A,e_{i_1},\ldots,e_{i_{m-k}} \right] ,
  $$
  so that by the genericity assumption
  $$
  \dim ( \PPA + \bigcap_{j\in J}\Pi_j )=k+m-k=m.
  $$
  It follows that
  \begin{align*}
    \dim( \PPA \cap \bigcap_{j\in J} H_j) &= \dim( \PPA \cap
    \bigcap_{j\in J} \Pi_j) \\
    &= \dim( \PPA ) + \dim(\bigcap_{j\in J} \Pi_j) - \dim( \PPA
    +\bigcap_{j\in J} \Pi_j)
    \\
    &= k + (m - \vert J \vert) - m
    \\
    &= \dim ( \PPA ) -\vert J \vert,
  \end{align*}
  so the subspaces $\PPA \cap \Pi_i$ are subspaces in $\PPA$ in
  general position.

  The second claim is elementary.
\end{proof}

\bigbreak

Thus, the number of orthants intersected by a generic $n$-dimensional
subspace of $\RR^m$ equals the number of cells determined by $m$
general-position, codimension-$1$ subspaces of $\RR^n$. Our next step
is to show that this value is given by
$$
Q(m,n) := 2 \left[
  \binom{m-1}0 + \binom{m-1}1 + \cdots + \binom{m-1}{n-1}
\right] .
$$
Figure~\ref{fig:q-table} provides a small table of $Q(m,n)$ values.


\begin{figure}[ht]
\begin{center}
  \begin{tabular}[h]{|c||rr|rr|rr|rr|}
    \hline
    & \multicolumn{8}{c|}{$n$} \\
    \hline
    $m$  & 1 & 2 & 3 & 4 & 5 & 6 & 7 & 8 \\
    \hline\hline
    1 & 2 &  2&  2 &  2 &  2 &  2 &  2 &  2
    \\
    2 & 2 & 4 & 4 & 4 & 4 & 4 & 4 & 4
    \\
    \hline
    3 & 2 & 6 & 8 & 8 & 8 & 8 & 8 & 8
    \\
    4 & 2 & 8 & 14 & 16 & 16 & 16 & 16 & 16
    \\
    \hline
    5 & 2 & 10 & 22  & 30 & 32 & 32 & 32 & 32
    \\
    6 & 2 & 12 & 32 & 52 & 62 & 64 & 64 & 64
    \\
    \hline
    7 & 2 & 14 & 44 & 84 & 114 & 126 & 128 & 128
    \\
    8 & 2 & 16 & 58 & 128 & 198 & 240 & 254 & 256 \\
    \hline
  \end{tabular}
\end{center}
\caption{A table of $Q(m,n)$ values.}
\label{fig:q-table}
\end{figure}

\begin{prop}
  \label{prop:Q-props}
  For positive integers $m$ and $n$ we have
  \begin{itemize}
  \item [(i)] $Q(m,1) = 2$,
  \item [(ii)] $Q(m,2) = 2m$,
  \item [(iii)] for $m \le n$, $Q(m,n)=2^m$  and so, in particular,
    $Q(1,n)=2$, and
  \item [(iv)] for $m,n\ge2$, $Q(m,n) = Q(m-1,n) + Q(m-1,n-1)$.
  \end{itemize}
\end{prop}

\begin{proof}
  Claims (i)--(iii) are elementary and (iv) is established by the
  following routine calculation:
  \begin{align*}
    Q(m-1,n) + Q(m-1,n-1) &=
    2\sum_{j=0}^{n-1} \binom{m-2}{j} + 2 \sum_{j=0}^{n-2}
    \binom{m-2}{j} \\
    &=
    \binom{m-2}0 + 2 \sum_{j=0}^{n-2}
    \left[
      \binom{m-2}{j+1} + \binom{m-2}{j}
    \right]
    \\
    &= 2 \binom{m-1}0 + 2 \sum_{j=0}^{n-2} \binom{m-1}{j+1}
    \\
    &=
    2 \sum_{j=0}^{n-1} \binom{m-1}j = Q(m,n).\qedhere
  \end{align*}
\end{proof}

\begin{lemma}
  \label{cell_count_lemma}
  Given $m$ subspaces of codimension $1$ in general position in a
  vector space $\VV$ of dimension $n$, the number of connected components
  of $\VV - \bigcup_{i=1}^m H_i$ depends only on $n$ and $m$ and is given
  by $Q(m,n)$.
\end{lemma}

For example, see Figure~\ref{fig:sliced-sphere}.

\begin{figure}[ht]
  \begin{center}
    \includegraphics[scale=0.5]{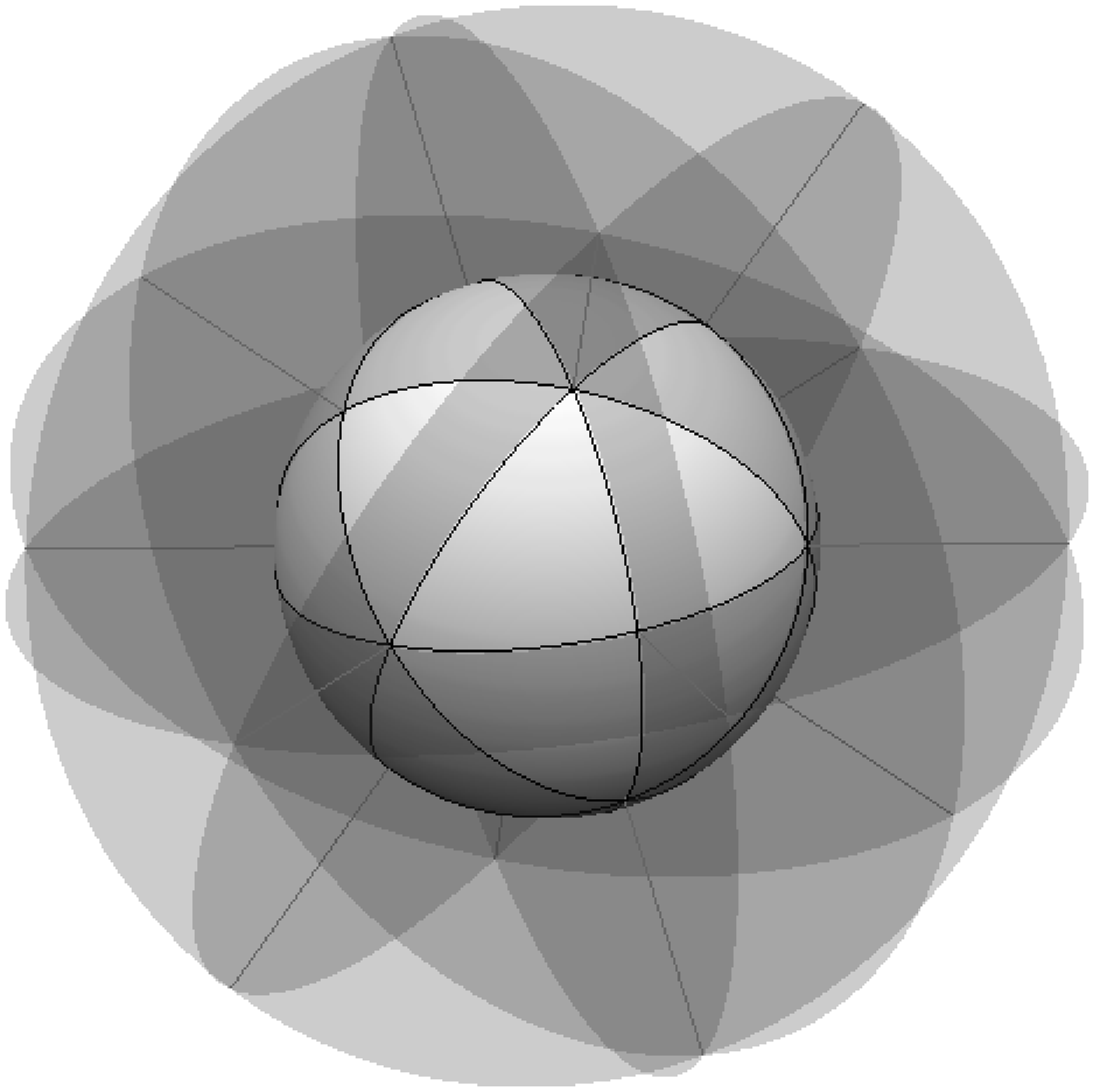}
  \end{center}
  \caption{Slicing $\RR^3$ with $4$ codimension-1, general-position
    subspaces gives $Q(4,3)=14$ regions.}
  \label{fig:sliced-sphere}
\end{figure}

\begin{proof}
  The result is easy to verify in case $m=1$ and in cases $n=1,2$, so we
  may assume $m,n\ge2$.

  It suffices, therefore, to show that the number of connected
  components satisfies the recursion formula in
  Proposition~\ref{prop:Q-props}.  We make the inductive assumption
  that the claim in the Lemma holds for $m-1$ subspaces (any $n$).


  Consider $m$ general-position, codimension-1 subspaces
  $\VV_1,\VV_2,\ldots,\VV_m$ of $\VV$. By induction, the first $m-1$
  of these cut $\VV$ into $Q(m-1,n)$ regions.

  Now consider the intersection of $\VV_1,\ldots,\VV_{m-1}$ with
  $\VV_m$. Let $\tilde\VV_i = \VV_i \cap \VV_m$. Note that
  $\tilde\VV_1,\ldots,\tilde\VV_{m-1}$ are codimension-1,
  general-position subspaces of $\VV_m$, and therefore divide $\VV_m$
  into $Q(m-1,n-1)$ regions.

  So the first $m-1$ subspaces $\VV_i$ divide $\VV$ into $Q(m-1,n)$
  regions and the addition of $\VV_m$ subdivides $Q(m-1,n-1)$ of those
  previous regions, giving a total of
  $$
  Q(m-1,n) + Q(m-1,n-1) = Q(m,n)
  $$
  regions, as required.
\end{proof}




Lemma~\ref{cell_count_lemma} and its proof are analogous to the
following well-known geometry problem \cite{buck}: Given $m$
hyperplanes (not necessarily subspaces) in general position in
$\RR^n$, the number of regions defined by these hyperplanes is
$$
\binom m0 + \binom m1 + \cdots + \binom mn
$$
which equals $\frac12 Q(m+1,n+1)$. We can show this connection
directly by the following geometric argument.

Consider $m$ general-position hyperplanes in $\RR^n$; these determine
$r$ regions. Extend $\RR^n$ to $\RP^n$ and consider this arrangement
of hyperplanes as sitting in $\RP^n$. When we do this some of the
regions ``wrap'' past infinity, so to maintain the same number of
regions $r$, we add the hyperplane at infinity to the arrangement.
See Figure~\ref{fig:slice-rp2}.
Recall that a point in $\RP^n$ corresponds to a line through the
origin in $\RR^{n+1}$ and a hyperplane in $\RP^n$ corresponds to a
codimension-1 subspace of $\RR^{n+1}$. Since the hyperplanes of the
original arrangement in $\RR^n$ (and $\RP^n$) are in general position,
these are now general position subspaces of $\RR^{n+1}$. Each region
of the arrangement of $m+1$ hyperplanes in $\RP^n$ splits into two
antipodal regions when viewed as subspaces of $\RR^{n+1}$. Thus, the
$r$ original regions yield $2r$ regions determined by $m+1$
codimension-1, general position subspaces of $\RR^{n+1}$, and so
$2r=Q(m+1,n+1)$.

\begin{figure}[ht]
  \begin{center}
    \includegraphics[scale=0.333]{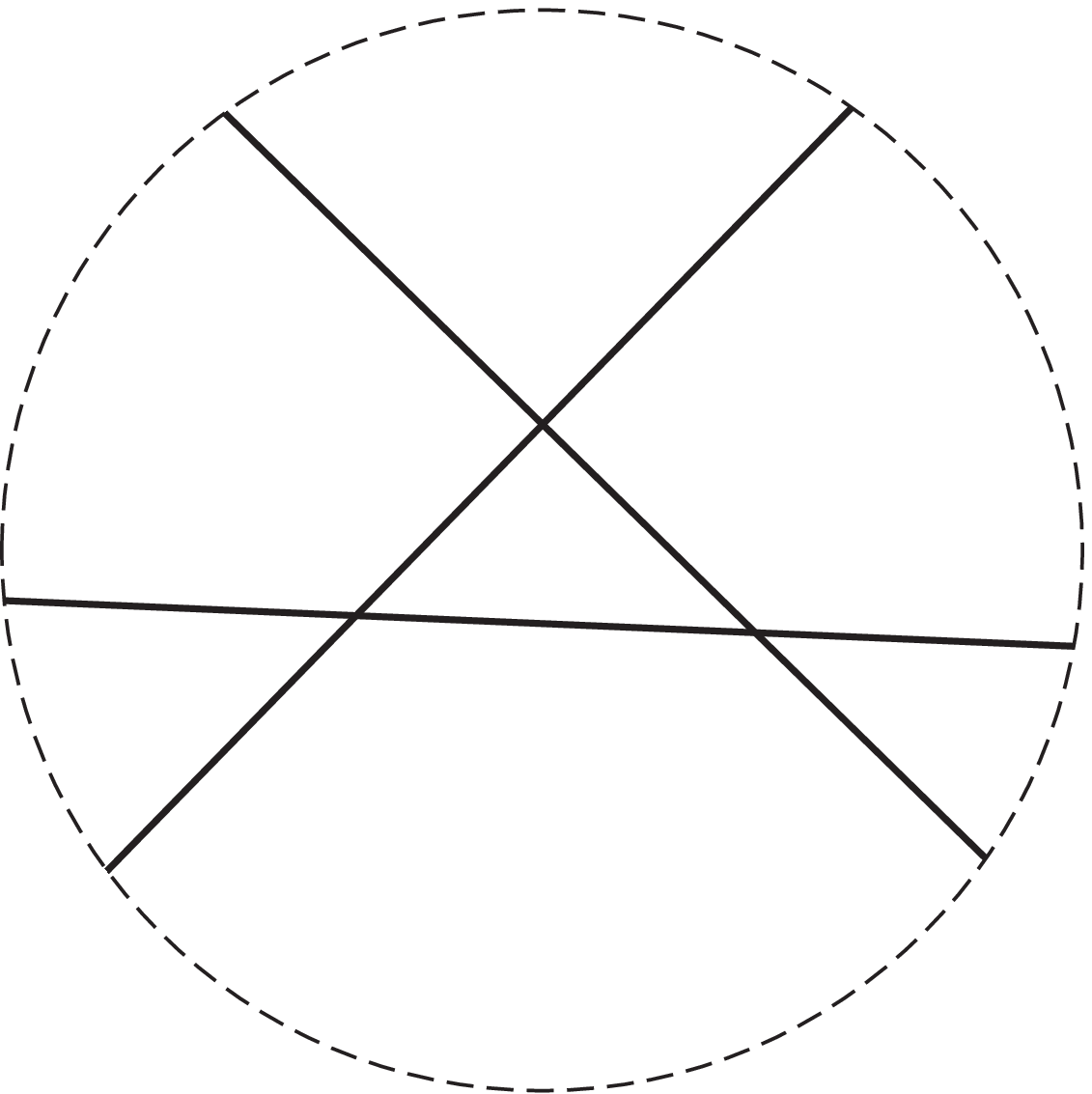}
  \end{center}
  \caption{Slicing the plane with three lines, or the projective plane
    with four lines, gives $\frac12 Q(4,3)=7$ regions.}
  \label{fig:slice-rp2}
\end{figure}

\bigbreak

The values $Q(m,n)$ have an additional geometric interpretation. For
example, the third column of the table in Figure~\ref{fig:q-table} is
sequence A014206 of \emph{The On-Line Encyclopedia of Integer
  Sequences} \cite{online-integers}; these numbers are the maximum
number of regions determined by $n$ circles in the plane; see
Figure~\ref{fig:circles}.
\begin{figure}[ht]
  \begin{center}
    \includegraphics[scale=0.3333]{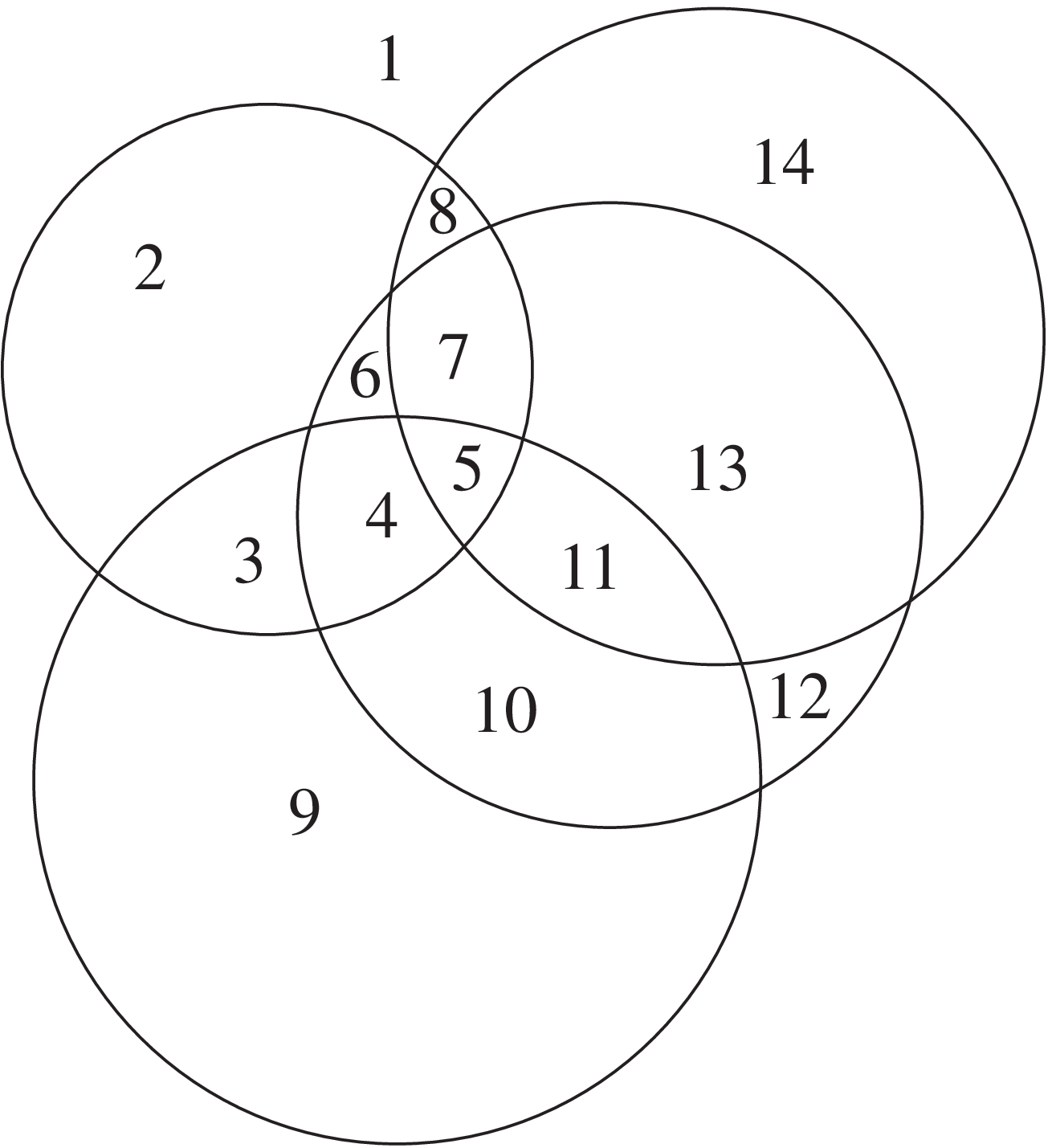}
  \end{center}
  \caption{Four circles in the plane determine $Q(4,3)=14$ regions.}
  \label{fig:circles}
\end{figure}
In general, the maximum number of regions determined by $m$ balls in
$\RR^{n-1}$ is $Q(m,n)$.

The various geometric interpretations of $Q(m,n)$ are collected in the
following result.

\begin{theorem}
  \label{thm:super-qnd}
  Let $m,n$ be positive integers. Then $Q(m,n)$ equals all of the
  following:
  \begin{enumerate}
  \item $2\sum_{k=0}^{n-1} \binom{m-1}{k}$.
  \item The number of regions defined by $m$ general-position,
    codimension-$1$ subspaces of $\RRn$.
  \item The number of orthants intersected by a generic
    $n$-dimensional subspaces of $\RR^m$.
  \item The maximum number of regions determined by $m$ balls in
    $\RR^{n-1}$.
  \item Twice the number of regions determined by $m-1$
    general-position hyperplanes in $\RR^{n-1}$.
  \item Twice the number of regions determined by $m$ general-position
    hyperplanes in $\RP^{n-1}$. \qed
  \end{enumerate}
\end{theorem}

\begin{corollary}
  \label{generic_payoff_matrix_corollary}
  Let $A$ be a generic $m \times n$ payoff matrix.  Then the number of
  orthants that $\PPA$ intersects is $Q(m,n).$
\end{corollary}

\begin{proof}
  Since a generic payoff matrix is full-rank with $n \leq m$ the
  dimension of $\PPA$ is $n.$ By Lemma
  \ref{generic_implies_general_position_lemma}, $\PPA \cap \Pi_1,
  \ldots,\PPA\cap \Pi_m$ are general-position, codimension-1
  subspaces of  $\PPA$, and the nonempty cells they
  define are in one-to-one correspondence with the orthants that
  $\PPA$ intersects. Using Lemma \ref{cell_count_lemma} the number
  of such cells is $Q(m,n).$
\end{proof}

\section{Random payoff matrices and the probability of arbitrage}

Much of the research in mathematical finance is focused on modeling
various markets and using these models to price various financial
instruments. Our point of departure is to put aside entirely the
consideration of \emph{real} markets, and consider mathematically
convenient \emph{random} payoff matrices.

We ask: \emph{What is the probability that a random $m\times n$ payoff
  matrix admits an arbitrage opportunity?} To make this question
precise, one must give a probability distribution on
matrices. However, for a wide range of natural probability
distributions,\footnote{For example, let the $mn$ entries be iid
  $N(0,1)$ random values.}  %
the answer is $Q(m,n)/2^m$. Here's the intuition. The $n$-dimensional
column space of a random matrix $A$ intersects $Q(m,n)$
orthants. Since all orthants look the same, the probability that
$\PPA$ intersects $\OO^+$ is $Q(m,n)/2^m$.
In other words, a `random' $n$-dimensional subspace of $\RR^m$
intersects $\OO^+$ with probability $Q(m,n)/2^m$. (See \cite{may+smith}
who consider random polytopes in $\RR^m$.)

We now make this intuition precise.

\medbreak

We consider payoff matrices that are generic with probability
one. This is not a difficult condition to attain as the next result
makes clear.

\begin{lemma}
  \label{lemma_generic_with_probability_one}
  Let $n\leq m$ and suppose $A$ be a random $m \times n$ payoff matrix
  whose entries have a continuous joint distribution. Then $A$ is
  generic with probability one.
\end{lemma}

\begin{proof}
  We identify $m \times n$ payoff matrices with points in $\RR^{mn}$
  in an obvious way, and the assumption says that for some function
  $f:\RR^{mn} \rightarrow [0,+\infty)$ we have
  $$
  P[ A \in G] = \int_G f(x) dx
  $$
  for $G \subseteq \RR^{mn}.$ Genericity for $A$ fails if
  $\det[A,e_{i_1},\ldots,e_{i_{m-n}}]=0$ for some choice of distinct
  indices $i_1,\ldots,i_{m-n}.$ Since there are only finitely many
  choices for such indices, we need only show that each set
  $$
  \left\{ A\in \RR^{mn} ~:~ \det[A,e_{i_1},\ldots,e_{i_{m-n}}]=0 \right\}
  $$
  has Lebesgue measure 0. Since the determinant in this expression
  is a non-constant polynomial in the variables $A_{ij}$ the result
  follows.
\end{proof}

\begin{corollary}
  Let $n\leq m$ and suppose $A$ be a random $m \times n$ payoff matrix
  whose columns are iid random $m$-vectors whose distribution is
  continuous.  Then $A$ is generic with probability one.
\end{corollary}

Observe that if $n \geq m$ and the entries of $A$ have a continuous
joint distribution, then with probability one $\PPA=\RR^m,$ in which
case $\PPA$ intersects every orthant, and in particular arbitrage
occurs with with probability one. For this reason, it is more natural
to focus on the case $n \leq m.$

\begin{dfn}
  An $m \times n$ random payoff matrix is said to be \emph{invariant
    under reflections} if its distribution is unaffected when any of
  its  rows is multiplied by $-1.$
\end{dfn}

For example, if the $mn$ entries of $A$ are chosen independently from
some probability distribution that is symmetric about $0$, then $A$ is
invariant under reflections.

Given $\delta\in \SSS_m$ we let $R_\delta$ denote the $m \times m$
diagonal matrix whose diagonal is $\delta.$

\begin{lemma}
\label{orthant_hitting_lemma}
Given an $m \times n$ payoff matrix $A,$ and $\delta,\omega \in \SSS_m$ we
have $\PPA \cap \OO_\delta \neq \emptyset$ if and only if
$\PP(R_\omega A) \cap \OO_{\omega\delta} \neq \emptyset.$

\end{lemma}

\begin{proof}
  This follows from the easy-to-verify fact that $Ax \in \OO_\delta$
  if and only $R_\omega Ax \in \OO_{\omega\delta}$.
\end{proof}

\begin{lemma}
  \label{equal_probability_lemma}
  If an $m \times n$ random payoff matrix $A$ is invariant under
  reflections then $\PPA$ intersects each of the $2^m$ orthants in
  $\RR^m$ with equal probability.
\end{lemma}

\begin{proof}
  Given $\omega, \delta \in \SSS_m,$ Lemma \ref{orthant_hitting_lemma}
  gives
  $$
  P\left[ \PPA \cap \OO_\delta \neq \emptyset \right] = P\left[
    \PP(R_\omega A) \cap \OO_{\omega\delta} \neq \emptyset \right].
  $$
  On the other hand, invariance under reflections guarantees that
  $R_\omega A$ has the same distribution as $A.$ Thus, this last
  probability equals
  $$
  P\left[ \PPA \cap \OO_{\omega\delta} \neq \emptyset \right].
  $$
  That is, $\PPA$ intersects $\OO_\delta$ and
  $\OO_{\omega\delta}$ with equal probability.  Since $\omega\delta$
  varies over all of $\SSS_m$ as $\omega$ varies over $\SSS_m$ the
  result follows.
\end{proof}

Our main result is the following.

\begin{theorem}
  \label{main_theorem}
  If $A$ is a random $m \times n$ payoff matrix that is generic with
  probability one and invariant under reflections, then $\PPA$
  intersects the orthant $\OO_\delta$ with probability $Q(m,n)/2^m$
  for any $\delta \in \SSS_m.$ In particular, such a matrix admits an
  arbitrage opportunity with probability $Q(m,n)/2^m$.
\end{theorem}

\begin{proof}
  Let $N(A)$ denote the number of orthants that $\PPA$ intersects.
  Since $A$ is generic with probability one, we apply Corollary
  \ref{generic_payoff_matrix_corollary} to conclude that $E\left[ N(A)
  \right] = Q(m,n).$ On the other hand $N(A)$ is a sum of indicator
  random variables
  $$
  N(A) = \sum_{\delta \in \SSS_m} I \left\{ \PPA \cap \OO_\delta
    \neq \emptyset \right\}
  $$
  so that using Lemma \ref{equal_probability_lemma}
  \begin{align*}
  Q(m,n) &= E\left[ N(A) \right] \\
  & = \sum_{\delta \in \SSS_m} E\left[
    I \left\{\PPA \cap \OO_\delta \neq \emptyset \right\} \right]
  \\
  &= \sum_{\delta \in \SSS_m} P\left[ \PPA \cap \OO_\delta \neq
    \emptyset \right]
  \\
  &= \vert \SSS_m \vert P\left[ \PPA \cap \OO^+
  \right] \\
  &= 2^m P\left[ \PPA \cap \OO^+ \right]
  \end{align*}
  and the result follows.
\end{proof}

In particular, if $A$ is a random payoff matrix whose entries are
independent and normally distributed with mean zero, then $A$ is
generic with probability one and invariant under reflections.

\begin{corollary}
  Let $n$ be a positive integer.  If $A$ is a random payoff matrix
  satisfying the conditions of Theorem~\ref{main_theorem} with $n$
  columns (investments) and $2n$ rows (scenarios) then the probability
  of arbitrage is $1/2.$
\end{corollary}

For large random payoff matrices satisfying the conditions of the
theorem we can use the central limit approximation~\citep{FellerII} to
calculate the probability of an arbitrage.

\begin{corollary}
  Let $p(m,n)$ denote the probability of an arbitrage for an $m \times
  n$ random normal payoff matrix. If $n_m$ is a sequence such that
  $$
  \lim_{m \rightarrow \infty}
  \frac{n_m - (m-1)/2}{\sqrt{(m-1)/4}} = x \in (-\infty,\infty)
  $$
  then
  $$
  \lim_{m \rightarrow \infty} p(m,n_m) = \Phi(x),
  $$
  where $\Phi$ denotes the standard normal cumulative distribution
  function.
\end{corollary}

\end{document}